\newif\iflong
\newif\ifshort

\longtrue

\iflong
\else
\shorttrue
\fi

\documentclass[cleveref,a4paper,USenglish,modulo,thm-restate]{lipics-v2021}
\iflong
\hideLIPIcs  %
\pdfoutput=1 %uncomment to ensure pdflatex processing (mandatatory e.g. to submit to arXiv)
\fi

\usepackage[utf8]{inputenc}
\usepackage{cite}
\usepackage{todonotes,xspace}
\usepackage[outercaption]{sidecap} 
\usepackage{hyperref}
\usepackage{boxedminipage}

\usepackage{environ}
\usepackage{booktabs}

\newcommand{\repeattheorem}[1]{%
  \begingroup
  \renewcommand{\thetheorem}{\ref{#1}}%
  \expandafter\expandafter\expandafter\theorem
  \csname reptheorem@#1\endcsname
  \endtheorem
  \endgroup
}

\NewEnviron{reptheorem}[1]{%
  \global\expandafter\xdef\csname reptheorem@#1\endcsname{%
    \unexpanded\expandafter{\BODY}%
  }%
  \expandafter\theorem\BODY\unskip\label{#1}\endtheorem
}

\DeclareMathOperator{\tww}{tww}
\DeclareMathOperator{\cw}{cw}
\DeclareMathOperator{\tw}{tw}
\newcommand{\cc}[1]{{\mbox{\textnormal{\textsf{#1}}}}\xspace}  %
\newcommand{\var}{\text{var}}
\newcommand{\cla}{\text{cla}}
\newcommand{\ones}{\text{ones}}
\newcommand{\ol}[1]{\overline{#1}}
\newcommand{\bigoh}{\mathcal{O}}
\newcommand{\W}[1][xxxx]{\text{\normalfont W[#1]}}
\newcommand{\paranp}{\text{\normalfont \cc{paraNP}}}
\newcommand{\NP}{\cc{NP}}

\newcommand{\FPT}{\cc{FPT}}

\newcommand{\BWMC}{\textup{\textsc{BWMC}}\xspace}
\newcommand{\BSAT}{\textup{\textsc{BSAT}}\xspace}
\newcommand{\prof}{\mathtt{Profiles}}

\newtheorem{property}{Property}[theorem]

\def\hy{\hbox{-}\nobreak\hskip0pt} 
\newcommand{\Card}[1]{|#1|}

\newcommand{\SB}{\{\,}%
\newcommand{\SM}{\mid}
\newcommand{\SE}{\,\}}%

\title{Weighted Model Counting with Twin-Width}

\author{Robert Ganian}{Algorithms and Complexity Group, TU Wien,
  Vienna,
  Austria}{rganian@ac.tuwien.ac.at}{https://orcid.org/0000-0002-7762-8045}{support
  by the Austrian Science Fund (FWF, project Y1329)}  \author{Filip
  Pokr\'yvka}{Masaryk University, Brno,
  Czechia}{xpokryvk@fi.muni.cz}{https://orcid.org/0000-0003-1212-4927}{supported
  by the Czech Science Foundation, project no. 20-04567S}
\author{Andr\'e Schidler}{Algorithms and Complexity Group, TU Wien,
  Vienna, Austria}{aschidler@ac.tuwien.ac.at}{}{} \author{Kirill
  Simonov}{Algorithms and Complexity Group, TU Wien, Vienna,
  Austria}{ksimonov@ac.tuwien.ac.at}{}{supported by the Austrian
  Science Fund (FWF, projects Y1329 and P31336)}  \author{Stefan
  Szeider}{Algorithms and Complexity Group, TU Wien, Vienna,
  Austria}{sz@ac.tuwien.ac.at}{https://orcid.org/0000-0001-8994-1656}{supported
  by the Austrian Science Fund (FWF, project P32441) and by the
  Vienna Science and Technology Fund (WWTF, project ICT19-065)}

\authorrunning{R.\ Ganian, F.\ Pokr\'yvka, A.\ Schidler, K.\ Simonov and S.\ Szeider}
\Copyright{R.\ Ganian, F.\ Pokr\'yvka, A.\ Schidler, K.\ Simonov and S.\ Szeider}



\category{} %

\supplement{The code and experimental results are available at \url{https://doi.org/10.5281/zenodo.6610819}.}%

\acknowledgements{The authors thank \'Edouard Bonnet for his helpful feedback regarding the relationship between signed twin-width and planar graphs.
}%

\nolinenumbers %

\ifshort
\EventEditors{Kuldeep S. Meel and Ofer Strichman}
\EventNoEds{2}
\EventLongTitle{25th International Conference on Theory and Applications of Satisfiability Testing (SAT 2022)}
\EventShortTitle{SAT 2022}
\EventAcronym{SAT}
\EventYear{2022}
\EventDate{August 2--5, 2022}
\EventLocation{Haifa, Israel}
\EventLogo{}
\SeriesVolume{236}
\ArticleNo{15}
\fi

\begin{document}
\maketitle              %

\begin{abstract}
  Bonnet et al.\ (FOCS 2020) introduced the graph invariant twin-width
  and showed that many NP-hard problems are tractable for graphs of
  bounded twin-width, generalizing similar results for other width
  measures, including treewidth and clique-width.  In this paper, we
  investigate the use of twin-width for solving the propositional
  satisfiability problem (SAT) and propositional model counting. We
  particularly focus on Bounded-ones Weighted Model Counting (BWMC),
  which takes as input a CNF formula $F$ along with a bound $k$ and
  asks for the weighted sum of all models with at most $k$ positive
  literals. BWMC generalizes not only SAT but also (weighted) model
  counting.

  We develop the notion of ``signed'' twin-width of CNF formulas and establish that BWMC is fixed-parameter tractable when parameterized by the certified signed twin-width of $F$ plus $k$.
  We show that this result is tight: it is neither possible to drop the bound $k$ nor use the vanilla twin-width instead if one wishes to retain fixed-parameter tractability, even for the easier problem SAT. Our theoretical results are complemented with an empirical evaluation and comparison of signed twin-width on various classes of CNF formulas.
\end{abstract}

\section{Introduction}

In many cases, it is not sufficient to determine whether a propositional formula is satisfiable, but we also need to determine the number of models. This \textsc{Model Counting} problem arises in several areas of artificial intelligence, among others in the
context of probabilistic reasoning\cite{BacchusDalmaoPitassi03,GomesSabharwalSelman09,SangBeameKautz05}, and is often studied in a weighted setting, where each literal has a weight and each model contributes a weight that is equal to the product of the weights of its literals. Here, we will go a step further and consider a natural generalization of these problems called \emph{Bounded-ones Weighted Model Counting} (\BWMC), where we are additionally provided a bound $k$ on the input and are only asked to count models with at most $k$ literals sets to true.

\textsc{Model Counting} is known to be \#P-complete already when all variables have the
weight~1~\cite{Valiant79b} and remains \#P-hard even for monotone 2CNF
formulas and Horn 2CNF formulas~\cite{Roth96}.  Hence standard
\emph{syntactical restrictions} do not suffice to achieve tractability, not even for this restricted case of \BWMC.

A more successful approach for tackling model counting problems is based on \emph{structural   restrictions}, which focus on exploiting the interactions between variables and/or clauses by considering suitable graph representations of the input formula.
  The two most popular graph representations used in this context are the primal graph and the incidence graph\footnote{Definitions are provided in Section~\ref{sec:prelims}.} (sometimes called the variable interaction and variable-clause interaction graphs, respectively)~\cite{SamerSzeider21}. Typically, one aims at identifying structural properties of these graphs---measured by an integer \emph{parameter} $k$---which can be exploited to obtain so-called \emph{fixed-parameter} algorithms for a considered problem, which are algorithms whose worst-case running time is upper-bounded by $f(k)\cdot n^{\bigoh(1)}$ for some computable function $f$ and inputs of size $n$. Within the broader context of \emph{parameterized complexity theory}~\cite{DowneyFellows13,Cyganbook}, we then say that the problem is \emph{fixed-parameter tractable} w.r.t.\ the considered parameter(s).
  
  The arguably most classical results that arise from this
  ``parameterized'' approach to propositional satisfiability and \textsc{Model Counting} are the
  fixed-parameter algorithms w.r.t.\ the treewidth of the
  primal and incidence graphs~\cite{GottlobScarcelloSideri02,SamerSzeider10,SamerSzeider21}. These were then
  followed by the fixed-parameter tractability of the problem w.r.t.\
  the signed clique-width~\cite{FischerMakowskyRavve06} and signed rank-width~\cite{GanianHlinenyObdrzalek13}, as
  well as other results which combine structural restrictions with
  syntactic ones~\cite{GanianRamanujanSzeider17b}; all of these results can be seen as a push
  in the overarching aim of identifying the ``broadest,'' i.e., most
  general properties of graphs that suffice for fixed-parameter
  tractability. It is worth noting that all of these
  algorithmic results can be adapted to also solve
  \BWMC.

Recently, Bonnet et al.~\cite{TWW_I} discovered a fundamental graph parameter called \emph{twin-width} that is based on a novel type of graph decomposition called a \emph{contraction sequence}. They showed that bounded
twin-width generalizes many previously known graph classes for which
first-order model checking (an important meta-problem in computational
logic) is tractable; most notably, it is upper-bounded on graphs of bounded clique-width as well as on planar graphs~\cite{TWW_I}. In this sense, it provides a ``common generalization'' of both of these (otherwise very diverse) notions. In spite of its recent introduction, twin-width has already become the topic of extensive research~\cite{TWW_II,TWW_III,TWW_IV,dreier2022twin}.

\bigskip\noindent\textbf{Contributions.}\quad
With this paper, we embark on investigating the utilization of
twin-width for the propositional satisfiability problem (SAT) and, more generally,
\BWMC. We begin by noting that, similarly to  the case of clique-width, it is impossible to exploit the ``vanilla'' notion of twin-width, even for SAT. Indeed, as will become clear in Section~\ref{sec:lowerb}, neither \BWMC nor SAT is fixed-parameter tractable when parameterized by the vanilla twin-width of their primal or incidence graph representations. 
Hence, inspired by previous work on SAT using clique-width~\cite{FischerMakowskyRavve06} and rank-width~\cite{GanianHlinenyObdrzalek13}, we develop a notion of \emph{signed twin-width} of CNF formulas that is based on the incidence graph representation along with new \emph{bipartite contraction sequences}. 

As our main algorithmic result, we establish the following:

\begin{reptheorem}{main}
\label{the:main}
\BWMC is fixed-parameter tractable when parameterized by $k$ plus the twin-width of a signed contraction sequence provided on the input.
\end{reptheorem}

\smallskip

We show that this result is essentially tight in the sense that  to retain fixed-parameter tractability, it is neither possible to replace the use of ``signed'' twin-width with vanilla twin-width as introduced by Bonnet et al.~\cite{TWW_I}, nor use the primal graph representation, nor drop $k$ as a parameter. Our results are summarized in Table~\ref{tab:th_results}.

\begin{table}[htbp]
  \begin{center}
    \begin{tabular}{@{}l@{\qquad}l@{\qquad}l@{\qquad}l@{}}\toprule
      \textbf{~} & \textbf{Signed twin-width} & \textbf{Vanilla twin-width} & \textbf{Primal twin-width} \\ \midrule
      \textbf{$k$ is parameter} & \FPT (Theorem~\ref{the:main}) & \W[1]-hard (Pr.~\ref{prop:mcc}) & \W[2]-hard (Pr.~\ref{prop:w2}) \\
      \textbf{$k$ is unrestricted} & \paranp-hard (Pr.~\ref{prop:paranp}) & \paranp-hard (Pr.~\ref{prop:paranp}) & \paranp-hard (Pr.~\ref{prop:paranp})\\
     \bottomrule
    \end{tabular}
  \end{center}
  \vspace{-0.2cm}
  \caption{Overview of results. The fixed-parameter tractability applies to \textsc{BWMC}, while all lower bounds hold already for the problem of deciding whether a model exists.}
  \label{tab:th_results}
\end{table}

Apart from establishing our main algorithmic result and the accompanying lower bounds highlighted above, we also prove that incidence graphs of bounded signed twin-width are a strictly more general class than both planar incidence graphs and incidence graphs of bounded signed clique-width. In fact, for the latter case, our proof also yields an improved bound on the vanilla twin-width for graphs of bounded clique-width compared to that of Bonnet et al.~\cite{TWW_I}. 

We complement our theoretical findings with a brief experimental
evaluation, where we compare the signed clique-width, signed twin-width, and treewidth of
CNF formulas. For this comparison, we
utilize SAT encodings that determine the exact value of these
parameters \cite{HeuleSzeider15,Parlak16,SchidlerSzeider22,SamerVeith09}. 
For signed twin-width, we had to adapt the respective encoding to signed graphs
and used the bipartiteness of the incidence graph to compute larger instances
than would be possible with the plain encoding.
Surprisingly, even though the signed twin-width could, in theory, be (at most a constant factor) larger than both treewidth and signed clique-width, in all of the experiments, it turned out to be smaller than these two other parameters.

\ifshort
\noindent
\emph{Statements where proofs are provided in the full version are marked with }$\star$.
\fi

\section{Preliminaries}
\label{sec:prelims}
We assume a basic knowledge of common notions used in graph theory~\cite{Diestel12}.
For an integer $n > 0$ we denote the set $\{1, \ldots, n\}$ by $[n]$.
All graphs considered in this paper are simple. 
For two vertices $u, v \in V(G)$ we denote by $uv$ the edge with endpoints $u$ and $v$.

Following the terminology of~\cite{TWW_I}, we call the following operation a \emph{contraction} of two vertices $u$ and $v$: introduce a new vertex $w$ into the graph whose neighborhood consists of all the neighbors of $u$ and $v$, and remove $u$ and $v$ from the graph. This definition of a contraction is distinguished from the more commonly used term ``edge contraction'' that corresponds to the same operation but requires $u$ and $v$ to be adjacent.

\smallskip
\noindent \textbf{Satisfiability and Weighted Model Counting.}\quad
We consider propositional formulas in conjunctive normal form (CNF),
represented as sets of clauses over a variable set $\var$.  That is, a \emph{literal} is a
(propositional) variable $x$ or a negated variable $\ol{x}$, a
\emph{clause} is a finite set of literals not containing a
complementary pair $x$ and $\ol{x}$, and a \emph{formula} is a finite set
of clauses. We use $\var(F)$ to denote the variables of $F$.

A \emph{truth assignment} (or \emph{assignment}, for short) is a
mapping $\tau:X\rightarrow \{0,1\}$ defined on some set $X$ of variables in a formula $F$.
We extend $\tau$ to literals by setting $\tau(\ol{x})=1-\tau(x)$ for
$x\in X$. 
An assignment
$\tau:X\rightarrow \{0,1\}$ \emph{satisfies} a formula $F$ if
every clause of $F$ contains a literal $z$ such that $\tau(z)=1$. A truth assignment $\tau:\var(F)\rightarrow
\{0,1\}$ that satisfies~$F$ is a \emph{model} of $F$, and let $M(F)$ be the set of all models of $F$. 

Let $w$ be a \emph{weight function} that maps each literal of $F$ to a real. The weight of an assignment $\tau$ is the product over the weights of its literals, i.e., $w(\tau)=\big(\Pi_{v\in\tau^{-1}(1)}(w(v)\big) \cdot \big( \Pi_{v\in\tau^{-1}(0)}(w(\neg v)\big) \big)$. Another property of assignments that will be useful for our considerations is the number of variables set to 1 by the assignment, which we define as $\ones(\tau)=|\tau^{-1}(1)|$. Our main problem of interest is defined as follows:

\begin{center}
\begin{boxedminipage}{0.98 \columnwidth}
\textsc{Bounded-ones Weighted Model Counting (BWMC)}\\[5pt]
\begin{tabular}{p{0.08 \columnwidth} p{0.8 \columnwidth}}
    Input: & A formula $F$ with a weight function $w$ and an integer $k$.\\
Task: & Compute $\sum_{\pi\in M(F) ~ \wedge ~ \ones(\pi)\leq k }w(\pi)$.
\end{tabular}
\end{boxedminipage}
\end{center}

We note that the \textsc{Weighted Model Counting} problem precisely corresponds to \BWMC\ when we set $k=|\var(F)|$ and to \textsc{Model Counting} when we additionally also set each literal to a weight of one. When describing our lower bounds, it will also be useful to consider a simpler decision problem called \textsc{Bounded-ones SAT} (\BSAT), which takes as input a formula $F$ and integer $k$ and asks whether the formula admits a satisfying assignment with at most $k$ variables set to 1. In other words, \BSAT is the restriction of \BWMC\ to the case where each literal has a weight of one and where we only need to decide whether the output is $\geq 1$ or $0$; this problem precisely corresponds to SAT when we set $k=|\var(F)|$.

The (\emph{signed}) \emph{incidence graph} of $F$, denoted here by $G_F$, is an edge-labeled graph defined as follows: $V(G_F)=F\cup \var(F)$ and $E(G_F)=\{ab~|~a\in F\wedge b\in \var(F) \wedge \{b, \ol{b}\}\cap a\neq \emptyset \}$. The edge set is partitioned into the set of \emph{positive} edges $E^+(G_F)=\{ab~|~a\in F\wedge b\in \var(F) \wedge b \in a \}$ and the set of \emph{negative} edges $E^-(G_F)=\{ab~|~a\in F\wedge b\in \var(F) \wedge \ol{b} \in a \}$. We will later also use the \emph{primal graph} of $F$ for comparison; the vertex set of the primal graph is the set of all variables, and two variables $a,b$ are connected by an edge if and only if there is at least one clause containing literals of both variables.

\smallskip
\noindent \textbf{Twin-width.}\quad
Twin-width was introduced by Bonnet, Kim, Thomass\'{e}, and Watrigant~\cite{TWW_I}; in what follows, we recall the basic concepts and notations introduced there.
A \emph{trigraph} is defined by a triple $(V(G), E(G), R(G))$ where $E(G)$ and $R(G)$ are both sets of edges with endpoints in $V(G)$, called (usual) black edges and \emph{red edges}. We call $G$ a \emph{$d$-trigraph} if the maximum degree in a subgraph induced by $R(G)$ is at most $d$. A trigraph is \emph{red-connected} if it is connected and remains connected after removing all black edges.

For a trigraph $G$ and vertices $u, v \in V(G)$, we define the trigraph $G / u,v$ obtained by \emph{contracting} $u, v$ into a single vertex $w$. Specifically, $V(G/u,v) = V(G) \setminus \{u, v\} \cup \{w\}$, $G - \{u, v\} = G/u,v - \{w\}$, and the edges incident to $w$ are as follows: for a vertex $x \in V(G / u, v) \setminus \{w\}$,
\begin{itemize}
    \item $wx \in E(G/u,v)$ if and only if $ux \in E(G)$ and $vx \in E(G)$,
    \item $wx \notin E(G/u,v) \cup R(G/u,v)$ if and only if $ux \notin E(G) \cup R(G)$ and $vx \notin E(G) \cup R(G)$,
    \item $wx \in R(G/u,v)$ otherwise.
\end{itemize}
We say that $G/u,v$ is a \emph{contraction} of $G$, and if both are $d$-trigraphs, $G/u,v$ is a \emph{$d$-contraction} of $G$. A graph $G$ is $d$-collapsible if there exists a sequence of $d$-contractions that contracts~$G$ to a single vertex. The minimum $d$ for which $G$ is $d$-collapsible is the \emph{twin-width} of $G$, denoted $\tww(G)$.

\smallskip
\noindent \textbf{Signed Clique-Width and Clique-Width Expressions.}\quad
In Section~\ref{sec:signed} we draw connections between twin-width and signed clique-width, and here we introduce the basic definitions for the latter. 
For a positive integer $k$, we let a \emph{$k$-graph} be a graph whose vertices
are labeled by $[k]$. For convenience, we consider a graph to be a $k$-graph with all vertices labeled by $1$. 
We call the $k$-graph consisting of exactly one vertex $v$ (say,
labeled by $i$) an initial $k$-graph and denote it by $i(v)$.

The (\emph{signed}) \emph{clique-width} of a signed graph $G$ is the smallest integer $k$ such that
$G$ can be constructed from initial $k$-graphs by means of iterative
application of the following four operations:
\begin{enumerate}
	\item Disjoint union (denoted by $\oplus$);
	\item Relabeling: changing all labels $i$ to $j$ (denoted by $\rho_{i\rightarrow j}$);
	\item Positive edge insertion: adding positive edges from each vertex labeled by $i$
	to each vertex labeled by $j$ ($i\neq j$; denoted by
	$\eta^+_{i,j}$);
	\item Negative edge insertion: adding negative edges from each vertex labeled by $i$
	to each vertex labeled by $j$ ($i\neq j$; denoted by
	$\eta^-_{i,j}$).	
\end{enumerate}
A construction of a $k$-graph $G$ using the above operations can be
represented by an algebraic term composed of $\oplus$,
$p_{i\rightarrow j}$, $\eta^+_{i,j}$ and $\eta^-_{i,j}$ (where $i\neq j$ and $i,j\in
[k]$). Such a term is called a $k$-expression defining $G$, and we often view it as a tree with each node labeled with the appropriate operation.
Conversely, we call the $k$-graph that arises from a \(k\)-expression its \emph{evaluation}.
The \emph{clique-width} of a signed graph $G$ is the smallest integer $k$ such that $G$
can be defined by a $k$-expression, which we then also call a \emph{clique-width expression} of \(G\).

\smallskip
\noindent \textbf{Parameterized Complexity.}\quad
Next, we give a brief and rather informal review of the most important
concepts of parameterized complexity. For an in-depth treatment of the
subject, we refer the reader to other sources
\cite{Cyganbook,DowneyFellows13,FlumGrohe06,Niedermeier06,SamerSzeider21}.
\ifshort ($\star$) \fi

The instances of a parameterized problem can be considered as pairs
$(I,k)$ where~$I$ is the \emph{main part} of the instance and $k$ is
the \emph{parameter} of the instance; the latter is usually a
non-negative integer.  A parameterized problem is
\emph{fixed-parameter tractable} (FPT) if instances $(I,k)$ of size $n$
(with respect to some reasonable encoding) can be solved
in time $f(k)n^c$ where $f$ is a computable function and $c$
is a constant independent of $k$. Such algorithms are called \emph{fixed-parameter algorithms}.

\iflong
To obtain our lower bounds, we will need the notion of a parameterized reduction.
Let $L_1$, $L_2$ be parameterized decision problems.
A \textit{parameterized reduction} (or fpt-reduction) from $L_1$ to $L_2$ is a mapping $P$ from instances of $L_1$ to instances of $L_2$ such that
\begin{enumerate}
  \item $(x,k)$ is a YES-instance if and only if $P(x,k)$ is a YES-instance,
  \item the mapping can be computed by a fixed-parameter algorithm w.r.t.\ parameter $k$, and
  \item there is a computable function $g$ such that $k'\leq g(k)$, where $(x',k')=P(x,k)$.
\end{enumerate}
\fi

All our lower bounds will be obtained already for the simpler \BSAT problem, allowing us to restrict our attention to decision problems. Since this is \NP-complete, the strongest form of parameterized intractability one can establish for a parameterization of \BSAT is \paranp-hardness, which means that the parameterized problem remains \NP-hard even for a fixed constant value of the considered parameter. A weaker notion of intractability is provided by the complexity classes \W[1] or \W[2]; while \W[2] is believed to be a superclass of \W[1], both \W[1]-hardness and \W[2]-hardness exclude fixed-parameter tractability under well-established complexity assumptions~\cite{Cyganbook}. It is perhaps worth noting that the whole W-hierarchy (which contains the complexity classes \W[1], \W[2],\dots,\W[P]) is itself defined through a weighted satisfiability problem that is similar in spirit to \BSAT~\cite{DowneyFellows13}.

\section{The Twin-Width of Signed Graphs}\label{sec:signed}
Motivated by the incidence graph of a formula, we define a \emph{signed graph} to be a graph~$G$ distinguishing two sets of edges $E^+(G)$ and $E^-(G)$ (positive and negative edges) over the same set of vertices $V(G)$. We assume these two sets of edges to be disjoint, i.e., $E^+(G)\cap E^-(G)=\emptyset$, so each pair of vertices $u,v\in V(G)$ is either positively-adjacent, negatively-adjacent or non-adjacent. A signed graph $G$ is bipartite if the set of vertices $V(G)$ can be partitioned into two independent sets.

We now proceed to define the twin-width for signed graphs based on
contractions, analogously to how twin-width is defined for graphs~\cite{TWW_I} defines. First, we have to incorporate red edges in a signed graph, representing ``errors'' that were created during contractions. Let $G$ be a \emph{signed trigraph} which contains, in addition to $E^+(G)$ and $E^-(G)$, an additional set $R(G)$ of edges called the \emph{red edges}. A signed trigraph $G'$ is a contraction of $G$ if there exist vertices $u,v\in V(G)$ and a vertex $w\in V(G')$ such that $G\setminus\{u,v\}=G'\setminus\{w\}$ and all vertices $x\in V(G')\setminus\{w\}$ satisfy the following:

\begin{itemize}
  \item $wx$ is a positive edge if both $ux$ and $vx$ are positive,
  \item $wx$ is a negative edge if both $ux$ and $vx$ are negative,
  \item $w$ and $x$ are non-adjacent if $u$ and $x$ are non-adjacent and $v$ and $x$ are non-adjacent,
  \item otherwise, $wx$ is a red edge.
\end{itemize}

A contraction sequence of a signed graph $G$ is a sequence of signed trigraphs $G=G_n,G_{n-1},\dots,G_1=K_1$, where $G_i$ has $i$ vertices and $G_i$ is a contraction of $G_{i+1}$. A contraction sequence of a signed graph $G$ is called a \emph{$d$-sequence} if each vertex in each signed graph of the sequence has red degree at most $d$. The twin-width of a signed graph is the minimal such $d$ over all contraction sequences of $G$; we call this the \emph{signed twin-width} of $G$. Later on, it will also be useful to compare this signed twin-width to the twin-width of the underlying unsigned graph for $G$; to avoid any ambiguities, we refer to this parameter as the \emph{unsigned twin-width} of $G$.

One useful way of viewing the contraction sequence is through vertex partitions. Let $G$ be a signed graph, and consider a partition $V'$ of $V(G)$. Intuitively, each set of vertices $X\in V'$ corresponds to a vertex in the contracted graph that represents vertices in $V(G)$ that were contracted into $X$. The set of positive, negative, and red edges over $V'$ is then defined as follows.

\begin{itemize}
  \item $X,Y\in V'$ are positively-adjacent, if each vertex $x\in X$ is positively-adjacent to each vertex $y\in Y$.
  \item $X,Y\in V'$ are negatively-adjacent, if each vertex $x\in X$ is negatively-adjacent to each vertex $y\in Y$.
  \item $X,Y\in V'$ are non-adjacent, if each vertex $x\in X$ is non-adjacent to each vertex $y\in Y$.
  \item $X,Y\in V'$ are red-adjacent otherwise.
\end{itemize}

In this way, each signed graph $G_i$ in a contraction sequence of $G$ can be treated as a vertex partition of $V(G)$, where each vertex $v\in V(G_i)$ is a subset of vertices of $V(G)$ that were contracted into $v$.

\smallskip
\noindent \textbf{Bipartite contraction sequences.}\quad
In this section we define bipartite sequences of contractions, which are a core feature in the definition of the signed twin-width of formulas. In this context, they are helpful as they allow us to keep variable and clause vertices separate.

\begin{definition}\label{def_bipartite}
  Let $G$ be a bipartite graph with vertex partition $A,B\subseteq V(G),$ $A\cap B = \emptyset$, $A\cup B = V(G)$.
  A $d$-sequence of $G$ is a bipartite $d$-sequence of $G$, if each contraction in the sequence contracts two vertices from the same part (either $A$ or $B$).
\end{definition}

Observe that each graph in a bipartite $d$-sequence is bipartite and a bipartite $d$-sequence of maximal length ends with $G_2$, where we can no longer contract.
At this point, we can formalize our parameter of interest: 

\begin{definition}
\label{def:signedtwf}
For a formula $F$ with signed incidence graph $G$, let the \emph{signed twin-width} of $F$ be the minimal $d$ such that $G$ has a bipartite $d$-sequence ending in a two-vertex graph.
\end{definition}

Below, we show that the restriction to bipartite contraction sequences is not only natural in the context of incidence graph representations of formulas but also ``essentially harmless'' in terms of the parameter's behavior.

\begin{lemma}\label{lemma_bipartite}
  Let $G$ be a bipartite signed graph $G$ and $G=G_n,\dots,G_1$ be a $d$-sequence. Then a bipartite $(d+2)$-sequence $G=G'_n,\dots,G'_2$ can be computed in linear time.
\end{lemma}

\begin{proof}
Suppose $G=G_n,G_{n-1},\dots,G_1 = K_1$ is a $d$-sequence of $n$-vertex bipartite graph $G$, with vertex partition $A,B$. Each vertex $u$ of $G_i$ is a set of vertices $u(G)$ of $G_n$ that contracted into $u$. Consider the following algorithm which constructs a bipartite $(d+2)$-sequence $G'_i$ along with a mapping $m$ of indices between these sequences, which is necessary due to the fact that there will not be a one-to-one correspondence between trigraphs in sequence $G'$ and in $G$; $m$ will map the index $j$ of the trigraph we are currently processing in $G'$ to a corresponding index $i$ in $G$.

Initially, we set $G'_n = G_n$. For each $G_i$ starting with $i:=n-1$ and ending with $i:=1$, denote $w$ the newly created vertex in $G_i$ and $u,v$ the vertices of $G_{i+1}$ that contracted to $w$.
The set of vertices $w(G)$ can be split into $w_A=w(G)\cap A$ and $w_B=w(G)\cap B$, and the sets $u_A,u_B,v_A,v_B$ are defined analogously. We can proceed depending on the number of empty sets among $u_A,u_B,v_A,v_B$.

\begin{itemize}
  \item Two sets are empty. Then both $u$ and $v$ belong to one part. If they are in the same part, it is safe to contract them as $w$ is also contained in one part. We append $G'_j$ to the contraction sequence, set $m(j)=i$, and proceed to the next trigraph by decrementing both $i$ and $j$.
  Otherwise a contraction is not possible; we do not change the sequence of $G'$ but continue processing the next graph in $G$, decrement $i$.

  \item One set is empty. Without loss of generality, assume $v_B$ is the empty set (if not, swap the labels between $u,v$ or $A,B$). In this case, we construct the graph $G'_j$ where $u_A,v_A$ have been contracted and append it to $G'$, set $m(j)=i$, and proceed to the next trigraph by decrementing both $i$ and $j$.

  \item All four sets are non-empty. We call this operation a \emph{double step}, since two contractions are required. Firstly, we contract the $A$-parts ($u_A$ and $v_A$) to construct a first trigraph $G'_j$. We add this trigraph to the sequence $G'$, set $m(j)=i$ and decrement $j$. Next, we proceed analogously for the new value of $j$: we contract the $B$-parts ($u_B$ and $v_B$), add this new trigraph to the sequence $G'$, and once again set $m(j)=i$. Finally, we proceed to the next trigraph by decrementing both $i$ and $j$.
\end{itemize}

Obviously, the time complexity of the algorithm described above is linear; we merely process the input sequence and at each graph in the sequence we add either zero, one or two graphs into the output sequence.

What remains to be shown is that $G'$ is a bipartite $(d+2)$-sequence. First, it is obviously a bipartite sequence, since at each step we contracted only $A$-parts and $B$-parts of vertices being contracted in sequence of $G$, and since $G_1$ is single vertex $x$, $G'_2$ contains two vertices $x_A$ and $x_B$.
Next, we turn our attention to the red degrees.

\begin{claim}\label{claim_vertex_red_degree_d}
  Let $j\in [n]$ such that $m(j)\neq m(j-1)$ (i.e., $G'_j$ is not an intermediate graph in a double step). Let $x\in V(G_{m(j)})$ be a vertex in a trigraph in the original sequence. Then $x_A$ has red degree at most $d$ in $G'_j\setminus x_B$, and conversely $x_B$ has red degree at most $d$ in $G'_j\setminus x_A$.
\end{claim}

\begin{claimproof}
  We prove the claim for $x_A$ only, since the situation for $x_B$ is fully symmetrical. Suppose $x_A$ has red degree more than $d$ in $G'_j\setminus x_B$. This means that there are at least $d+1$ vertices in the $B$-part of $G'_j\setminus x_B$ that are red-adjacent to $x_A$ (there are no red edges inside the $A$-part and inside the $B$-part).
  Each vertex  $b_i$ in the red neighborhood of $x_A$ is then a $B$-part of some vertex of $G_{m(j)}$; let us denote these vertices as $u_i$. Now, each $u_i$ is red-adjacent to $x$ in $G_{m(j)}$, because $b_i$ is a red neighbor of $x_A$, and by contracting more vertices with $x_A$ and $u_i$, the red edge does not disappear.
  That makes the red-degree of $x$ at least $d+1$ in $G_{m(j)}$, which contradicts the input sequence of $G$ being a $d$-sequence.
\end{claimproof}

As an immediate corollary of Claim~\ref{claim_vertex_red_degree_d}, we obtain unless $G'_j$ is an intermediate graph in a double step, all vertices in $G'_j$ have red degree at most $d+1$.

It remains to check intermediate graphs produced in the middle of double steps.
In this case, there are two contractions in the sequence $G'$, while there is a single corresponding contraction in $G$. Thus, it might happen that after the first contraction in the $A$-part, we introduce a new red edge to a vertex that will be contracted in the second contraction in the $B$-part. For all other vertices, Claim~\ref{claim_vertex_red_degree_d} applies analogously. Furthermore, the same argument also upper-bounds the red degree of the vertices in the $A$-part that are being contracted, but here we must consider the graph $G'_j\setminus\{u_B,v_B\}$ which yields a total upper-bound of $d+2$ on all red degrees in $G'$.
\end{proof}

We remark that twin-width also admits a more general definition over matrices, where its definition on graphs coincides with the use of the matrix definition on the adjacency matrix of the graph~\cite{TWW_I}. In this context, the signed twin-width of a formula $F$ could equivalently be stated as the twin-width of a ``signed'' bipartite adjacency matrix of the (signed) incidence graph, where the edge labels in the graph are represented as different values in the matrix.

\smallskip
\noindent
\textbf{Comparison to Existing Measures.}\quad
Bonnet et al.~\cite{TWW_I} showed that any proper minor-closed graph class has bounded twin-width. Therefore, in particular, planar graphs and graphs of bounded genus have bounded twin-width. While, strictly speaking, this result considers only graphs and not signed graphs, the subsequent work~\cite{TWW_II} provides a way to lift these bounds to signed graphs. $(\star)$

\iflong
Specifically, Theorem 12 of~\cite{TWW_II} implies that for any sparse hereditary class of bounded twin-width (where ``sparse'' here means not containing $K_{t, t}$ as a subgraph for some constant $t$), it holds that the adjacency matrix of any graph in this class is \emph{$d$-grid free}~\cite{TWW_II}, for some constant $d$, which intuitively means that the matrix does not contain a $d\times d$ ``non-uniform'' submatrix. After replacing the ones with one of the two new values (representing the signs in a signed graph), the matrix still does not contain a $d\times d$ ``non-uniform'' submatrix (here, one formalizes this using the notion of \emph{$d$-mixed free} matrices, since here the matrix is non-binary). By Theorem 10 of~\cite{TWW_I}, the twin-width of a $d$-mixed free matrix is bounded, and this precisely matches the signed twin-width of the corresponding signed graph.
\fi
\begin{proposition}[Corollary of Theorem~10~\cite{TWW_I} and Theorem~12~\cite{TWW_II}]
    Let $G$ be a signed graph and $G'$ be the corresponding unsigned graph. If $G'$ has twin-width $d$ and does not contain $K_{t, t}$ as a subgraph then the signed twin-width of $G$ is at most $2^{2^{\bigoh(dt)}}$.
\end{proposition}

\begin{corollary}
\label{cor:planar}
The class of all planar signed graphs has bounded signed twin-width.
\end{corollary}

However, in the above it is crucial that the graph class in question is sparse. For arbitrary graphs one should not expect to derive a signed twin-width bound from a twin-width bound of the unsigned graph. For example, cliques have bounded twin-width, but by assigning signs to edges of the clique one may obtain an arbitrary binary structure. On the other hand, the bounds on twin-width based on other width measures can be extended to bounds on signed twin-width based on a \emph{signed} version of the corresponding width measure. In what follows we show that bounded signed clique-width implies bounded signed twin-width.

\iflong
\begin{proposition}
\fi
\ifshort
\begin{proposition}[$\star$]
\fi
    Let a graph $G$ have a clique-width of $d$. Then $G$ has a twin-width of at most $2d$.
    Moreover, if $G$ is a signed graph with a signed clique-width of $d$, the signed twin-width of $G$ is at most $2d$.
\end{proposition}
\ifshort
\begin{proof}[Proof Sketch]
    We process the clique-width expression in a leaves-to-root fashion, during which we maintain the following invariant: all the vertices in the subexpression that have the same label are contracted into a single vertex. Thus, after processing each operation, the subexpression corresponds to at most $d$ vertices in the current graph in the contraction sequence, and no edges that go from the vertices of the subexpression to the outside, are red. Originally, no vertices are contracted, and by the time the whole expression is processed all vertices are contracted into one. It now remains to consider the three types of operations.

    In an edge insertion node, no contractions occur as no vertices change labels. Moreover, the edges introduced correspond to exactly one edge in the current contraction sequence graph, and that edge is black since only vertices that have the same label were contracted together. In a relabeling node $i \to j$, we contract the vertices that represent labels $i$ and $j$ in the current subexpression. No outside edges become red as afterwards in the expression these vertices are treated identically, and inside the subexpression there are at most $d$ vertices, thus the red degree in the subgraph corresponding to the subexpression does not exceed $d$.
Finally, in a disjoint union node, we consecutively contract the vertices corresponding to the same label in the two disjoint parts. Analogously to the previous case, no outside edges become red. Since the whole subexpression corresponds to at most $2d$ vertices at this point, the red degree is at most $2d$ in each of the intermediate graphs.
\end{proof}
\fi
\iflong
\begin{proof}
    We start with showing the first part of the statement. Let $G$ be a graph and $\sigma$ be a $d$-expression defining $G$. We will iteratively construct a $2d$-contraction sequence of $G$ by processing the operations of $\sigma$ in a bottom-up way. More specifically, let us define a \emph{monotone subforest} of $\sigma$ to be a subforest of the expression tree of $\sigma$ that contains all the leaf nodes, and no vertex with a child outside of the subforest. We say that a monotone subforest $F$ is \emph{defined} by nodes $\tau_1$, \ldots, $\tau_\ell$ of $\sigma$ if $F$ is the union of the subtrees of $\tau_1$, \ldots, $\tau_\ell$. We will simultaneously construct a sequence $F_1$, \ldots, $F_t$ of monotone subforests of $\sigma$, a contraction sequence $G_n$, \ldots, $G_1$ of $G$, and a mapping $\pi : [t] \to [n]$ that establishes a correspondence between the two sequences. Intuitively, operations in $\sigma$ correspond to different number of contractions, thus the mapping $\pi$ is required to keep track of the corresponding indices in two sequences. Crucially, we construct the sequences in a way that satisfies the following property.

    \begin{property}\label{property:1}
        Let $j \in [t]$, $F_j$ be defined by nodes $\tau_1$, \ldots, $\tau_\ell$. Then for each $i \in [\ell]$, vertices of $G$ in $\tau_i$ that have the same label are contracted into a single vertex of $G_{\pi(j)}$ that has no other vertices of $G$ contracted into it, and these are all vertices of $G_{\pi(j)}$.
    \end{property}

    As we will see next, one big advantage of a construction satisfying Property~\ref{property:1} is that the latter ensures that no edge added outside of the subforest $F_j$ in $\sigma$ is red in $G_{\pi(j)}$, so all graphs in the contraction sequence corresponding to some $F_j$ by $\pi$ have low red degree.

    \begin{claim}\label{claim:prop_low_deg}
        Let $j \in [t]$, and let $F_j$ satisfy Property~\ref{property:1}. Then $G_{\pi(j)}$ is a $d$-trigraph.
    \end{claim}
    \begin{claimproof}
        Observe that for each node $\tau_j$, the vertices introduced in the subexpression of $\tau_j$ correspond to a subgraph on at most $d$ vertices in $G_{\pi(j)}$, by Property~\ref{property:1}. Moreover, the edges of $G_{\pi(j)}$ introduced outside of these subexpressions connect all of the vertices that have the same label in the same manner, thus they all remain black in $G_{\pi(j)}$. All the remaining edges are inside the subgraphs corresponding to $\tau_j$, and since there are at most $d$ vertices in each, the red degree does not exceed $d$.
    \end{claimproof}

    Now, we show how the sequences are constructed. We start with setting $G_n$ to $G$, $F_1$ to be the monotone subforest containing precisely the leaves of $\sigma$, and $\pi(1) = n$. $F_1$ is defined by the set of leaves of $\sigma$, so each vertex of $G$ has its own node among those, thus by $G_{\pi(1)} = G_n = G$ Property~\ref{property:1} holds for $F_1$. We proceed by constructing inductively $F_{i + 1}$ from $F_i$ for each $i$ until we obtain $F_t = \sigma$. We may assume w.l.o.g. that $\sigma$ uses only one label at the root node, which immediately implies that by Property~\ref{property:1} the final element in the contraction sequence is the single-vertex graph. Thus it only remains to show that on each step we can extend the contraction sequence such that (i) Property~\ref{property:1} holds for the new monotone subforest $F_{i + 1}$ and (ii) each of the newly added to the contraction sequence graphs has its red degree below $2d$. For the graphs that directly correspond to elements of $\{F_i\}$ via $\pi$ the latter holds automatically by Claim~\ref{claim:prop_low_deg}.

    In order to construct $F_{i + 1}$, we pick an arbitrary node $\tau$ of $\sigma$ such that $\tau \notin F_i$, but all of its children belong to $F_i$. It is easy to see that such a node exists since $F_i \ne \sigma$ and $F_i$ contains all leaves of $\sigma$. We set $F_{i + 1} = \sigma[V(F_i) \cup \{\tau\}]$, clearly $F_{i + 1}$ is a monotone subforest of $\sigma$. Now, consider separately the cases depending on the type of $\tau$. 

    \noindent$\tau = \rho_{x \to y}(\tau')$. Set $\pi(i + 1) = \pi(i) + 1$, we add one new graph $G_{\pi(i) + 1}$ to the contraction sequence. Specifically, by Property~\ref{property:1} in $G_{\pi(i)}$ there is one vertex where all vertices introduced in $\tau'$ that have the label $x$ are contracted, and the same holds for the label $y$.  Construct $G_{\pi(i) + 1}$ by contracting these two vertices. Clearly, Property~\ref{property:1} holds since the vertices that had the label $x$ in $\tau'$ have now the label $y$, and nothing changed for the remaining vertices. Since no intermediate graphs were added, this case is complete.

    \noindent$\tau = \tau^{(1)} \oplus \tau^{(2)}$. For each label $j \in [d]$ such that there are vertices with this label both in $\tau^{(1)}$ and $\tau^{(2)}$, we introduce a new graph in the sequence by contracting the corresponding vertices in $\tau^{(1)}$ and $\tau^{(2)}$ given by Property~\ref{property:1}. We set $\pi(i + 1)$ to be the index of the final graph in the contraction sequence after processing the label $d$. By construction, Property~\ref{property:1} holds for $G_{\pi(i + 1)}$. Analogously to the proof of Claim~\ref{claim:prop_low_deg}, we observe that red edges are constrained to be inside the subgraphs of the subexpressions that define $F_{i + 1}$. Since the size of each such subgraph is at most $d$ by Property~\ref{property:1}, except for the subgraph that corresponds to $\tau$, which contains at most $2d$ vertices in every step of the new part of the contraction sequence, the red degree in each of the added graphs is at most $2d$.

        \noindent$\tau = \nu_{x, y}(\tau')$. We do not add new steps in the contraction sequence, and we set $\pi(i + 1) = \pi(i)$. Since the set of vertices and their labels does not change between $F_i$ and $F_{i + 1}$, Property~\ref{property:1} immediately holds for $F_{i + 1}$.
\end{proof}
\fi

Moreover, twin-width strictly dominates clique-width: for example, the $n \times n$ grid has constant twin-width, but unbounded clique-width as $n$ grows~\cite{TWW_I}. The same holds for signed versions: signed clique-width clearly remains unbounded while the twin-width bound of Bonnet et al.~\cite{TWW_I} holds irrespectively of edge signs.

\begin{proposition}[Theorem 4,~\cite{TWW_I}]
    For every positive integers $d$ and $n$, the $d$-dimensional $n$-grid has twin-width at most $3d$.
    The same holds for the signed twin-width of any orientation of the grid.
\end{proposition}
\begin{proof}
    It suffices to observe that \cite{TWW_I} shows the bound for the case where all edges of the grid are red, thus the bound immediately transfers to grids with black edges where signs are arbitrary.
\end{proof}

In fact, even mim-width---a more general parameter than clique-width---is unbounded on $n \times n$ grids. In general, however, twin-width and mim-width are incomparable as, e.g., interval graphs have mim-width one and unbounded twin-width~\cite{TWW_II}. While mim-width has also been used to solve variants of SAT~\cite{SaetherTV14}, it does not yield fixed-parameter algorithms for these problems.
\section{Fixed-Parameter Algorithm for BWMC Parameterized by Twin-Width}

For a formula $F$ with signed incidence graph $G$, recall that the \emph{signed twin-width} of $F$ be the minimal red degree over all bipartite contraction sequences of $G$. This section is dedicated to our main technical contribution, which is a fixed-parameter algorithm for \textsc{WMC} parameterized by $k$ plus the signed twin-width of the input formula. 
We note that since a fixed-parameter algorithm for computing contraction sequences is not yet known, we will adopt the assumption that such a sequence is provided in advance.

\repeattheorem{main}

\iflong
\begin{proof}
\fi
\ifshort
\begin{proof}[Proof Sketch]
\fi
We begin by invoking Lemma~\ref{lemma_bipartite} on the input sequence, which constructs a signed bipartite $(\tww(G)+2)$-sequence $G_n,\dots,G_{2}$ of $G$ in linear time. The core of the proof will be a dynamic programming procedure which will proceed along this sequence where, on a high level, for each graph $G_i$ we will compute a record that will allow us to provide an output for the \textsc{WMC} problem once we reach $G_2$. To this end, we will dynamically compute records for $G_n$, $G_{n-1}$,\dots, $G_2$. Our records will consist of a mapping from so-called \emph{profiles} to reals, where each profile will correspond to a set of assignments of $\var(F)$ which ``behave the same'' on the level of $G_i$.

Formally, a profile for a signed trigraph $G_i$ is a tuple of the form $(T,P,M,\ell,Q)$, where:
\begin{itemize}
\item $T\subseteq V(G_i)$ that induce a red-connected subgraph of $G_i$ of size at most $k(d^2+1)$, partitioned into $T_{\cla}\subseteq V_{\cla}(G_i)$ and $T_{\var}\subseteq V_{\var}(G_i)$,
\item $P\subseteq T_{\var}$ where $|P|\leq k$,
\item $M\subseteq P$,
\item $\ell \leq k$, and
\item $Q\subseteq T_{\cla}$.
\end{itemize}

Intuitively, our records will only need to store information about parts of $G_i$ which are red-connected subgraphs, since---as we will see later---the uniformity of black edges allows us to deal with them without dynamic programming. Furthermore, since the number of ``ones'' is bounded by $k$, it will suffice to store information only about local parts of each red-connected subgraph, and for each such subgraph we will consider a separate $T$. For each fixed $T$, we will then keep track of all possible variable-vertices in $T_{\var}$ where positive assignments occur (via $P$), which of these also contain negative assignments (via $M$), how many ones we have used up so far in this choice of $T$ (via $\ell$), and which clause-vertices in $T_{\cla}$ are already satisfied by variables in $T_{\var}$ (via $Q$).

To formalize this intuition, we will use the notion of \emph{realizability}. Consider a profile $(T,P,M,\ell,Q)$ of $G_i$, and let $S^T_{\var}$ be the set of all variables in $V_{\var}(G_n)$ which are contracted to $T_{\var}$, similarly let $S^T_{\cla}$ be the set of all clauses in $V_{\cla}(G_n)$ which are contracted to $T_{\cla}$. We say that $(T,P,M,\ell,Q)$ is \emph{realizable} by an assignment $\nu$ of $S^T_{\var}$ if and only if:
\begin{description}
    \item[R1] for each vertex $u\in P$, at least one variable $v\in V(G_n)$ collapses to $u$ and $\nu(v)=1$,
    \item[R2] for each vertex $u\in M$, at least one variable $v\in V(G_n)$ collapses to $u$ and $\nu(v)=0$,
    \item[R3] for each vertex $u\in P\setminus M$ all variables $v\in V(G_n)$ collapsing to $u$ satisfy $\nu(v)=1$,
    \item[R4] exactly $\ell$ variables $v\in S^T_\var$ satisfy $\nu(v)= 1$.
    \item[R5] if $c\in Q$, then every clause of $S^T_\cla$ collapsing to clause-vertex $c$ is satisfied by $\nu$.
    \item[R6] if $c\notin Q$, then at least one clause of $S^T_\cla$ collapsing to $c$ is unsatisfied by $\nu$.
\end{description}

For each $G_i$, let $\prof(G_i)$ be the set of all profiles of $G_i$. For each profile $\tau=(T,P,M,\ell,Q)\in \prof(G_i)$, let $\alpha(\tau)=\{\nu:S^T_\var\rightarrow \{0,1\}~|~ \tau\text{ is realizable by }\nu\}$ be the set of all assignments which are, intuitively, captured by this profile. Next, we will use $w(\tau)$ to denote the sum of the weights of all these assignments; formally, $w(\tau)=\sum_{\nu\in \alpha(\tau)}w(\nu)$.
Then the \emph{record} for $G_i$, denoted $R_{G_i}$, is the mapping $\prof(G_i)\rightarrow \mathbb{R}$ which maps each profile $\tau\in \prof(G_i)$ to $w(\tau)$.

Denote $t=k(d^2+1)$ as an upper-bound of $|T|$. We start by upper-bounding the size of the records. %
Specifically, we show that $|\prof(G_i)|\leq s_i$ where $s_i$ is defined as $i(d^{2t-2}+1) {{t}\choose{k}} 2^{k+t}(k+1)$. \iflong This follows from bounding the possible choices for each component of the profile and multiplying the bounds. \fi
Using \cite[Lemma 8]{TWW_III}, the number of possible red-connected sets $T$, where $T\subseteq V(G_i)$ and $|T|\leq t$ is upper-bounded by $i(d^{2t-2}+1)$.
There are at most ${t}\choose{k}$ possible choices for the set $P$ since $|P|\leq k$ and $P\subseteq T_\var \subseteq T$. Since $M\subseteq P$, there are $2^k$ possible choices of $M$. For $\ell$, we have $k+1$ possibilities because $0\leq \ell\leq k$. For the last component $Q$, there are at most $2^t$ possible choices since $Q\subseteq T_\cla\subseteq T$.
 
Moreover, for $G_n$ we can construct $R_{G_n}$ as follows: for each $v\in V_\var(G_n)$, we will have two profiles $(\{v\},\{v\},\emptyset,1,\emptyset)$ and $(\{v\},\emptyset,\emptyset,0,\emptyset)$, while for each $v\in V_\cla(G_n)$ we will have a single profile $(\{v\},\emptyset,\emptyset,0,\emptyset)$. The first two profiles will be mapped to $w(\{v\mapsto 1\})$ and $w(\{v\mapsto 0\})$, respectively, while the last profile will be mapped to $0$.

Once we obtain the record for $G_2$, we can output the solution for the \textsc{WMC} instance as follows. First, we observe that if $G_2$ does not contain a red edge, then \emph{every} variable in the input formula $F$ occurs in \emph{every} clause of $F$ in the same way, which implies that $|\cla(F)|\leq 1$ and the instance is trivial. So, without loss of generality we can assume that $G_2$ contains a red edge, and let $\prof'(G_2)$ be the restriction of $\prof(G_2)$ to tuples where the first component $T$ contains precisely two vertices. From the definition of the records, it now follows that $\sum_{\pi\in M(F) ~ \wedge ~ \ones(\pi)\leq k }w(\pi) \quad = \quad \sum_{\tau\in \prof'(G_2)}R_{G_2}(\tau)$.

It remains to show how to compute $w(\tau)$ for each profile $\tau = (T, P, M, \ell, Q)$ of $G_i$ assuming that the record for profiles in $G_{i+1}$ is already computed. There are two possible cases, depending whether $G_i$ is the result of contracting clause vertices or variable vertices of $G_{i+1}$.
In both cases, denote by $x,y\in V(G_{i+1})$ the vertices that are contracted into $z\in V(G_i)$. Let $\tau = (T,P,M,l,Q)$ be a profile for $G_{i+1}$. If $z\not\in T$ then $\tau$ is also a valid profile for $G_{i + 1}$, and we conclude by setting $R_{G_i}(\tau) =R_{G_{i + 1}}(\tau) = w(\tau)$.
Thus, in the following we assume that $z\in T$. %

Let $T_1,\ldots, T_m$ be the red-connected components of $(T\setminus \{z\}) \cup \{x,y\}$ in $G_{i+1}$. Note that $m\leq d+2$ since each component contains either $x$, $y$, or a red neighbor of $z$.
First, let us observe that one of the following two cases must occur: either for every $T_j$ obtained for the current choice of $T$ it holds that $|T_j|\leq k(d^2 + 1)$ (which we call the \emph{Standard Case}), or $|T| = k(d^2 + 1)$ and $m = 1$ and $T_1 = T \setminus\{z\} \cup \{x, y\}$ (referred to as the \emph{Large-Profile Case}).

\smallskip
\noindent \textbf{The Standard Case.}\quad
Here, for each $T_j$ there is at least one and at most ${{|T_j|}\choose{k}} 2^{k+|T_j|}(k+1)$-many profiles in $\prof(G_{i+1})$, and we branch over all such profiles for each $T_j$. We obtain a set of profiles $\tau_1$, \ldots, $\tau_m$ where for each $j \in [m]$, $\tau_j = (T_j, P_j, M_j, \ell_j, Q_j)$ is a profile of $G_{i + 1}$ as $|T_j| \le k(d^2 + 1)$.
We only proceed if the following \emph{consistency conditions} between $\tau$ and $\tau_1, \ldots, \tau_m$ hold. First, we state conditions that are defined outside of $x$, $y$, $z$.
\begin{description}
    \item[CC1] $P \setminus \{z\} = \bigcup_{i = 1}^m P_j \setminus \{x, y\}$.
    \item[CC2] $M \setminus \{z\} = \bigcup_{i = 1}^m M_j \setminus \{x, y\}$.
    \item[CC3] $\ell = \sum_{i = 1}^m \ell_j$.
    \item[CC4] $Q \setminus \{z\} = (\bigcup_{j = 1}^m Q_j \setminus \{x, y\}) \cup B$, where $B$ is the set of clauses in $T_\cla \setminus\{z\}$ satisfied by $T_\var \setminus \{z\} \cup \{x, y\}$ via a black edge in $G_{i + 1}$. Specifically, a clause $c \in T_\cla \setminus \{z\}$ belongs to $B$ if there is a variable $v \in \bigcup_{j = 1}^m P_j$ such that $vc \in E^+(G_{i+1})$, or there is a variable $v \in \bigcup_{j = 1}^m (T_j)_\var \setminus (P \setminus M)$ such that $vc \in E^-(G_{i+1})$.
\end{description}
Moreover, we put special conditions depending on the type of the contraction that obtained $G_i$. Recall that since we consider only bipartite contraction sequences, $x$, $y$, $z$ are all vertices of the same type, either variables or clauses. Let $a, b \in [m]$ be such that $x \in T_a$, $y \in T_b$ ($a$ and $b$ are not necessarily distinct). 
\begin{description}
    \item[CC5] If $x, y, z$ are variable vertices, $z \in P \setminus M$ if and only if $x \in P_a \setminus M_a$ and $y \in P_b \setminus M_b$; $z \notin P$ if and only if $x \notin P_a$ and $y \notin P_b$; otherwise $z \in M$.
    \item[CC6] If $x, y, z$ are clause vertices, $z \in Q$ if and only if $x \in Q_a$ or is satisfied by $T_\var$ via a black edge in $G_{i + 1}$, and $y \in Q_b$ or is satisfied by $T_\var$ via a black edge in $G_{i + 1}$.
\end{description}
We say that $\tau_1$, \ldots, $\tau_m$ are \emph{consistent} with $\tau$ if $T_1$, \ldots, $T_m$ are defined as above, and all of \textsf{CC1--6} hold. Finally, we set
\begin{equation}
R_{G_i}(\tau) = \sum_{\substack{\tau_1, \ldots, \tau_m\\\text{consistent with } \tau}} \prod_{j = 1}^m w(\tau_j).
    \label{eq:wtau}
\end{equation}
Since for each $i \in [m]$, $w(\tau_j)$ is already computed in $R_{G_{i + 1}}(\tau_j)$, we can indeed compute the sum above.

\smallskip
\noindent \textbf{The Large-Profile Case.}\quad
Here, we do not have profiles on $T_1$ in $\prof(G_{i+1})$, as $|T_1| = |T|+1 > k(d^2+1)$. However, we can compute the record for $\tau$ using suitable profiles for $G_{i+1}$ of smaller size. We elaborate below.

We branch over the choice of a ``profile'' $\tau_1 = (T_1, P_1, M_1, \ell_1, Q_1)$ consistent with $\tau$. (Strictly speaking, a profile must have smaller size, but the consistency conditions are defined irrespectively of that.) It is easy to see that there are only constantly many consistent profiles as it only remains to determine how the vertices $x$ and $y$ belong to the sets $P_1$, $M_1$ (if variables are contracted) or $Q_1$ (if clauses are contracted). Let $v$ be the vertex with the maximum distance over the red edges from the set $P_1$ in $T_1$.
By this choice, $T'=T_1\setminus \{v\}$ remains red-connected. We consider profiles of $G_{i+1}$ of the form $(T', P_1,M_1,\ell, Q')$, where  $Q'$ is one of the suitable subsets of $Q_1 \setminus \{v\}$ that will be specified later.

If $v$ is a clause vertex, we only consider $Q' = Q_1\setminus\{v\}$.
We check whether $v$ is satisfied by a black positive edge from $P_1$, or by a negative black edge from $T'_\var\setminus(P_1\setminus M_1)$, or by a red edge.  In the latter case, we check explicitly whether variables that are contracted into the red neighborhood of $v$ satisfy all clauses contracted to $v$ when set to zero. \iflong We will show later that the red edge to $v$ ends in a variable vertex fully assigned to 0 by any assignment realizing the profile.\fi
Next, we check whether all clauses of $v$ are satisfied if and only if $v\in Q_1$. If this equivalence holds, we add to $R_{G_i}(\tau)$ the value $R_{G_{i+1}}(\tau')$, where $\tau'=(T',P,M,\ell,Q\setminus\{v\})$. If the checks above fail for all $\tau_1$, then the profile $\tau$ is not realizable.

If $v$ is a variable vertex, we go through all possible subsets $Q'\subseteq Q_1$ and check whether the profile $\tau'=(T', P,M,\ell, Q')$ is realizable. If so, then we also check whether all clauses in $Q_1\setminus Q'$ are satisfied and all clauses not in $Q_1$ are not satisfied by setting the variables of $v$ to 0. If all conditions hold, we add $R_{G_{i+1}}(\tau')$ to $R_{G_i}(\tau)$ computed so far.

At this point, it remains to argue the correctness of the computation in both cases 
\ifshort ($\star$) \fi
 and provide the claimed bound on the running time.

\iflong
\smallskip
\noindent\textbf{Correctness of the Standard Case.}\quad
 We now show that the value of $R_{G_i}(\tau)$ computed in \eqref{eq:wtau} is equal to $w(\tau)$. Specifically, we show a bijection between assignments $\pi$ of $S^{T}_\var$ that realize $\tau$, and tuples of assignments $(\pi_1, \ldots, \pi_m)$ where for each $i$, $\pi_j$ acts on $S^{T_j}_\var$ and realizes $\tau_j$, for some tuple of profiles $\tau_1$, \ldots, $\tau_m$ consistent with $\tau$.

First, let $\pi$ be an assignment of $S^{T}_\var$ that realizes $\tau = (T, P, M, \ell, Q)$. We construct profiles $\tau_1$, \ldots, $\tau_m$, and assignments $\pi_1$, \ldots, $\pi_m$ in the following way.
Observe that $S^T_\var$ is a disjoint union $\bigsqcup_{i = 1}^m S^{T_j}_\var$, for each $i \in [m]$ define $\pi_j$ to be the restriction $\pi\raisebox{-.5ex}{$|$}_{S^{T_j}_\var}$. Define $\tau_j$ to be the unique profile on $T_j$ such that $\pi_j$ realizes $\tau_j$, for each $i \in [m]$. It remains to show that the profiles $\tau_1$, \ldots, $\tau_m$ defined in this way are consistent with $\tau$.
\begin{description}
    \item[CC1--3] Observe that the sets $T_j \setminus \{x, y\}$ are disjoint and give in union the set $T \setminus \{z\}$. Moreover, for each $v \in S^{T_j}_\var$ by construction it holds that $\pi_j(v) = \pi(v)$. Thus for each $X \in T \setminus\{z\} = \bigcup_{i = 1}^m T_j \setminus\{x, y\}$, treated as a subset of $V(G)$, the assignment $\pi$ on $X$ coincides with $\pi_j$. Since the inclusion of $X$ in $P$, $P_j$, $M$, $M_j$ depends only on the assignment of the variables in $X$, we immediately get that $P \setminus \{z\} = \bigcup_{i = 1}^m P_j \setminus \{x, y\}$ and $M \setminus \{z\} = \bigcup_{i = 1}^m M_j \setminus \{x, y\}$. For the same reason, $\ell = \sum_{i = 1}^m \ell_j$. Therefore, \textsf{CC1--3} hold.
    \item[CC4] Analogously, the sets $(T_j)_\cla \setminus \{x, y\}$ are disjoint and give in union the set $T_\cla \setminus \{z\}$.  We first show that $\bigcup_{j = 1}^m Q_j \setminus \{x, y\}\subset Q \setminus \{z\}$. Let $c$ be a clause that collapses to $Q_j \setminus \{x, y\}$ and is satisfied by a variable $v \in S^{T_j}_\var$ in $\pi_i$. Since $S^{T_j}_\var \subset S^T_\var$ and $\pi(v) = \pi_i(v)$, $c$ is satisfied by $\pi$ as well. Therefore, for each $q \in Q_j \setminus \{x, y\}$, $q \in T \setminus \{z\}$ and all clauses of $q$ are satisfied by $\pi$, thus $q \in Q \setminus \{z\}$.
        Second, similarly $B \subset Q \setminus \{z\}$ since for each $q \in B$, every clause $c \in q$ is satisfied by a variable in $S^T_\var$ in some assignment $\pi_i$, and for each variable $\pi$ picks the same assignment.

        Finally, we show $Q \setminus \{z\} \subset \bigcup_{j = 1}^m Q_j \setminus \{x, y\} \cup B$. Let $q \in T_\cla \setminus\{z\}$ be such that every clause in $q$ is satisfied by $\pi$. If there is a variable $v \in S^T_\var$ such that $v$ collapses to $X \in T_\var \cup \{x, y\} \setminus \{z\}$ in $G_{i + 1}$, $\pi(v) = 1$, and there is a black edge $Xq \in E^+(G_{i + 1})$, then $q \in B \setminus \{z\}$ (analogously if $\pi(v) = 0$ and $Xq \in E^-(G_i)$). Otherwise, no black edge in $G_{i + 1}$ satisfies a clause in $q$, thus every clause $c \in q$ is satisfied via a red edge in $G_{i + 1}$. Consider the red-connected component $T_j$ of $T \setminus \{x, y\}$ that contains $q$. In $G_{i + 1}$, the red neighborhood of $q$ in $T_j$ is the same as in $T$ and the assignments $\pi$ and $\pi_j$ coincide on $S^{T_j}_\var$, thus $q \in Q_j \setminus \{z\}$.
   \item[CC5] Let $x$, $y$, $z$ be variable vertices, otherwise there is nothing to show. Let $z \in P \setminus M$, then each variable in $z$ is assigned to true in $\pi$. Since $z = x \cup y$, each variable in $x$ is assigned to true in $\pi_a$ and each variable in $y$ is assigned to true in $\pi_b$, thus $x \in P_a \setminus M_a$, $y \in P_b \setminus M_b$. The same holds in converse: if $x \in P_a \setminus M_a$ and $y \in P_b \setminus M_b$, for each variable $v \in x \cup y = z$, $\pi(z) = 1$, and thus $z \in P \setminus M$ since $\pi$ realizes $\tau$. Now let $z \notin P$, then no variable in $z$ is assigned to true by $\pi$. Analogously to the above, this is equivalent to $x \notin P_a$, $y \notin P_b$.
   \item[CC6] Let $x$, $y$, $z$ be clause vertices, and let $z \in Q$. Then each clause $c \in z$ is satisfied by $\pi$, meaning that each clause in $x$ and $y$ is satisfied by $\pi$. If there is a black edge from $T_\var$ to $x$ in $G_{i + 1}$ that satisfies $x$ by $\pi$, then the respective vertex $w \in T_j$ satisfies $x$ by $\pi_j$, then $w \in M_j$ or either $w \in P_j \setminus M_j$ or $w \notin P_j$ depending on the direction of the edge. If no black edge satisfies $x$ in $G_{i + 1}$, then every clause of $x$ is satisfied by a red edge to $T_\var$. Since all red edges from $x$ to $T_\var$ are preserved in $T_a$ and $\pi$ agrees with $\pi_a$ on $T_a$, $x$ must belong to $Q_a$. In either case, the forward direction of \textsf{CC6} is satisfied for $x$, and $y$ is treated analogously. In the other direction, if $x \in Q_a$ then each clause of $x$ is satisfied by the partial assignment $\pi_a$ via red edges, otherwise there is a black edge in $G_{i + 1}$ that satisfies all clauses of $x$ simultaneously by some variable in $S^T_\var$. In any case, every clause of $x$ is satisfied by $\pi$, and analogously for $y$.
\end{description}

In the other direction, consider profiles $\tau_1$, \ldots, $\tau_m$ consistent with $\tau$ and assignments $\pi_1$, \ldots, $\pi_m$ where $\pi_i$ acts on $S^{T_i}_\var$ and realizes $\tau_i$ for each $i \in [m]$. Let $\pi$ be the union of the assignments $\pi_1$, \ldots, $\pi_m$, it remains to show that $\pi$ realizes $\tau$. In what follows we verify the properties \textsf{R1--6}.
\begin{description}
    \item[R1--3] For a vertex $w \neq z$ this follows immediately from \textsf{CC1--3}. Let $x$, $y$, $z$ be vertex vertices. The multiset $\pi(z)$ is a union of $\pi(x)$ and $\pi(y)$, and thus \textsf{CC5} ensures that \textsf{R1--3} also hold for $z$.
    \item[R4] \textsf{CC3} immediately  implies that exactly $\ell = \sum_{j = 1}^m \ell_j$ variables in $S^T_\var$ are set to true.
    \item[R5--6] Let $c \in Q$, $c \ne z$. Let $c \in T_j$, by \textsf{CC4}, either $c \in Q_j$ or $c \in B$, i.e. $c$ is satisfied by a black edge from $T_\var \setminus \{z\} \cup \{x, y\}$. In the former case, the fact that $\pi_j$ realizes $\tau_j$ implies that every clause in $c$ is satisfied by $\pi_j$ in $T_j$ by \textsf{R4}. In the other case, w.l.o.g. let the satisfying edge $vc \in E^+(G_{i + 1})$, $v \in T_h$ for some $h \in [m]$ and $v \in P_h$. Since $\pi_h$ realizes $\tau_h$, by \textsf{R1} at least one variable in $v$ is set to true by $\pi_h$ and so by $\pi$, therefore every clause of $c$ is satisfied by this variable in $\pi$. The case $vc \in E^-(G_{i + 1})$ is analogous.

        On the other hand, let $c \ne z$ be not in $Q$, and let $c \in T_j$ for some $j \in [m]$. Then neither $c \in Q_j$ nor $c \in B$. The latter implies that no black edge from $T_\var \setminus \{z\} \cup \{x, y\}$ satisfies $c$, so the assignment $\pi$ does not satisfy any clause of $c$ via a black edge, thus the clauses of $c$ can only be satisfied via red edges to $T \setminus \{x, y\}$ in $G_{i + 1}$. By construction of $T_1$, \ldots, $T_m$, these red edges actually have their endpoints in $T_j$. Therefore, since $c \notin Q_j$, at least one of the clauses in $c$ is not satisfied by $\pi_j$. Since by the above no vertex outside of $T_j$ can satisfy a clause in $c$, the same clause is not satisfied by $\pi$ in $T$.

        Finally, let $c = z$. The arguments proceed similarly to the above except that now $c$ is the union of the clauses in $x$ and $y$, and we use \textsf{CC6} to argue about the status of the clauses. Let $z \in Q$, then both for $x$ and $y$ all the clauses are either satisfied by the partial assignment $\pi_a$ or $\pi_b$ on the respective profile, or all the clauses are satisfied by a black edge. On the other hand, if $z \notin Q$, either $x$ or $y$ is not satisfied by any black edge, and has at least one clause that is not satisfied internally by the partial assignment. Since all the red edges to $T_\var$ are actually in $T_a$ or $T_b$, respectively, this implies that $\pi$ also does not satisfy this clause.
\end{description}

Clearly, the two mappings above form a bijection, thus indeed the assignments $\pi$ that realize $\tau$ are in a one-to-one correspondence with tuples of the assignments $(\pi_1, \ldots, \pi_m)$ where for each $j \in [m]$, $\pi_j$ realizes $\tau_j$ and the profiles $\tau_1$, \ldots, $\tau_m$ are consistent with $\tau$.
Moreover, since in this bijection $\pi$ is a union of $\pi_1$, \ldots, $\pi_m$, it holds that $w(\pi) = \prod_{j = 1}^m w(\pi_j)$.
Thus,
\begin{multline*}
\sum_{\pi \text{ realizes } \tau} w(\pi) =\sum_{\substack{\tau_1, \ldots, \tau_m\\\text{consistent with } \tau}} \sum_{\pi_1 \text{ realizes } \tau_1} \cdots\sum_{\pi_m \text{ realizes } \tau_m} \prod_{j = 1}^m w(\pi_j)\\
= \sum_{\substack{\tau_1, \ldots, \tau_m\\\text{consistent with } \tau}} \prod_{j = 1}^m \left(\sum_{\pi_j \text{ realizes } \tau_j} w(\pi_j)\right) = \sum_{\substack{\tau_1, \ldots, \tau_m\\\text{consistent with } \tau}} \prod_{j = 1}^m w(\tau_j).
\end{multline*}
Since we compute our record $R_{G_i}(\tau)$ as the latter, we have indeed that $w(\tau) = R_{G_i}(\tau)$, 
which shows the correctness of the computation of the record.

\smallskip
\noindent \textbf{Correctness of the Large-Profile Case.}\quad
Since $|T_1|>k(d^2+1)$, we know that the distance from $v$ to the nearest vertex of $P_1$ is at least $3$ as $k(d^2+1)$ is the maximal number of vertices of distance at most $2$ from a set of size at most $k$ in a $d$\hy degree graph. First, let $v$ be a clause vertex. Every variable $u$ red-adjacent to $v$ can not be in $P_1$. This means, every assignment realizing profile $\tau_1$ sets all vertices contracted to $u$ to 0.
If $v$ is a vertex variable, its red\hy distance from $P_1$ is at least $4$ since the graph is bipartite. Therefore, each clause vertex $u$ red\hy adjacent to $v$ is at red-distance at least $3$ from $P_1$ (or $v$ would be closer to $P_1$ which is a contradiction). This means $u$ is not red\hy adjacent to $P_1$. Also all assignments realizing $\tau_1$ set to 0 all variables contracted to $v$ as  $v\not\in P_1$. This way, if $w\in Q_1$ is not yet fully satisfied in a smaller profile without $v$, all its clauses needs to be satisfied by $v$ being assigned to 0 (some of the clauses contracted to $w$ might be already satisfied via a red edge, but such an edge ends in a variable that is always set to 0, so no ambiguity in the assignment). The sets $P_1$ and $M_1$ might be only black-adjacent, so they would either satisfy all clauses contracted to $w$ or none of them. 

\fi

\noindent\textbf{Running time.} The number of profiles for $G_i$ is bounded by $s_i$, and $s_i \leq s_n$. Let $t=k(d^2+1)$ be the upper bound on the size of a profile. To compute the record for a single profile in $G_i$, we iterate through all possible tuples of consistent profiles.%
In the Large-Profile Case, we need to perform checks over all vertices contracted to a clause or a variable vertex, which can be done in time $\bigoh(n^2)$. This could happen at most $\bigoh(2^t)$ times, since we might iterate through all possible subsets of $Q$ in the profile. Thus we obtain the upper bound of $\bigoh(2^t n^2)$ for processing a single profile.
In the Standard Case, the number of tuples of consistent profiles is upper-bounded by $f(k,d)={{t + 1}\choose{k}} \cdot 2^{k+t + 1} \cdot (k+1)^{d+2}$. This holds since $T_1$, \ldots, $T_m$ is fixed for a particular profile $\tau$ of $G_i$, and it only remains to decide the sets $P_j$, $M_j$, $Q_j$, and the values $\ell_j$. There are at most ${t + 1}\choose{k}$ choices for all sets $P_j$ as all of the profiles are disjoint and contain at most $t + 1$ vertices, and at most $k$ of them can belong to any of $P_j$ in total. After fixing $P_j$, there are at most $2^k$ choices for the sets $M_j$ as for each $j \in [m]$, $M_j \subset P_j$. At most $2^{t + 1}$ choices exist for deciding which clauses be part of the sets $Q_j$. Finally, for each $j \in [m]$, $0 \le \ell_j \le k$, so the number of choices for $\ell_j$ is upper-bounded by $(k + 1)^{d + 2}$.

Hence, $R_{G_i}$ can be computed in time $\bigoh(s_{i} \cdot (f(k,d) + 2^t n^2)) = \bigoh(s_n\cdot (f(k, d) + 2^t n^2))$.
Since the length of the contraction sequence is $n-2\leq n$, this results in the overall running time bound of $\bigoh(n\cdot s_n \cdot (f(k,d) + 2^t n^2)) = n^4\cdot d^{\bigoh(kd^2)}$.
\end{proof}

\section{Tightness}
\label{sec:lowerb}

In this section, we show that our main result (Theorem~\ref{the:main}) is tight, in the sense that it is not possible to strengthen any of the parameterizations if one wishes to retain fixed-parameter tractability. Our hardness
results even hold for \BSAT, which merely asks whether a
given CNF formula has a satisfying assignment that sets at most $k$
variables to True.

\begin{enumerate}
\item \BWMC parameterized by $k$ plus the twin-width of the primal graph is $\W[2]$\hy hard
  (Proposition~\ref{prop:w2}).
\item \BWMC parameterized by the signed twin-width
  alone is \paranp\hy hard (Proposition~\ref{prop:paranp}).
\item \BWMC parameterized by $k$ plus the unsigned
  twin-width is $\W[1]$\hy hard (Proposition~\ref{prop:mcc}).
\end{enumerate}

\iflong
We now provide the proofs of these claims.
\fi
\ifshort
The proofs of these claims are provided in the full version.
\fi
\iflong
\begin{proposition}
\fi
\ifshort
\begin{proposition}[$\star$]
\fi
\label{prop:w2}
  \BSAT (and hence \BWMC) parameterized by $k$ is
  $\W[2]$\hy hard, even when restricted to formulas whose primal graphs have twin-width zero and even assuming that an optimal contraction sequence
  is provided on the input.
\end{proposition}

\iflong
\begin{proof}
  The $\W[2]$\hy hardness follows by a straightforward parameterized
  reduction from the following problem, which is known to well-known
  to be $\W[2]$\hy complete~\cite{FlumGrohe06}.
  \begin{quote}
    \textsc{Hitting Set}
  
    \emph{Instance:} A collection of subsets $S_1,\dots,S_m$ of a
    finite set $U$ and an integer $k$.
  
    \emph{Parameter:} The integer $d$.
  
    \emph{Question:} Is there a set $H\subseteq \bigcup_{i=1}^m$ of
    size $\leq k$ such that $H\cap S_i \neq \emptyset$ for all
    $1\leq i \leq m$ ($H$ is a hitting set of $S_1,\dots,S_m$).
  \end{quote}
  We consider $F=\{S_1,\dots,S_m\}$ as a CNF formula, and observe that
  satisfying assignments $\tau$ of $F$ with $\ones(\tau)\leq k$ are in
  one-to-one correspondence with the hitting sets of
  $S_1,\dots,S_m$. Furthermore, we add the  clause $U$;
  this clause does not affect the models of $F$ in any way, but it ensures that the primal graph of $G$ is a complete graph, and so it has twin-width zero. 
\end{proof} 
\fi

\iflong
\begin{proposition}
\fi
\ifshort
\begin{proposition}[$\star$]
\fi
\label{prop:paranp}
  \BSAT (and hence \BWMC) parameterized by the signed
  twin-width is \paranp\hy hard, even assuming that an optimal bipartite contraction sequence
  is provided on the input. The same also holds if we replace the signed twin-width with the vanilla twin-width of the incidence graph (i.e., where signs are ignored), or the twin-width of the primal graph.
\end{proposition}
\iflong
\begin{proof}
  When $k$, the bound on number of variables set to true, is not a
  parameter, then we can set $k$ to the number of variables, so the
  bound has no effect and we are left with the standard satisfiability problem
  SAT. The \paranp\hy hardness follows from the fact that planar
  signed graphs have constant twin-width by Corollary~\ref{cor:planar} (or, when considering the other two parameters, by earlier work on vanilla twin-width~\cite{TWW_I}),
  and SAT remains \NP-hard when restricted to planar instances~\cite{Lichetenstein82}. 
\end{proof}
\fi

 To establish the third claim, we first show that an edge subdivision of a clique has twin-width bounded by the size of the original clique. Note that we do not assume any additional properties of this subdivision---each edge could be subdivided an arbitrary number of times. We remark that this result provides an upper bound that complement the asymptotic lower bounds developed in earlier work on twin-width~\cite[Section 6]{TWW_II}.

\iflong
\begin{lemma}
\fi
\ifshort
\begin{lemma}[$\star$]
\fi
\label{lemma_clique_subdivision}
  Let $d\geq 2$ and $G$ be a graph obtained by an arbitrary sequence of subdivisions of the edges in the complete graph $K_d$. Then the twin\hy width of $G$ is at most $d-1$.
\end{lemma}

\iflong
\begin{proof}
  If we show that there is a contraction sequence, such that the degree of any vertex (i.e., the number of red and black edges incident to that vertex) in every graph in this sequence is at most $d-1$, we are finished as this is also a bound on the red degree.
  
  Divide vertices of $V(G)$ into sets $K$ and $S$, where $K$ are vertices of the original clique $K_d$, and $S=V(G)\setminus K$ is the set of vertices that subdivided the edges of $K_d$. Each vertex of $K$ has degree $d-1$, each vertex of $S$ has degree $2$, so each vertex of $V(G)$ has degree at most $d-1$.
  
  Pick vertex $v\in K$, such that it has a neighbor $u\in S$. Degree of $u$ is $2$, degree of $v$ is $d-1$. Contract $u$ and $v$ into new vertex $w$. Degree of vertex $w$ is $d-1$, because it is sum of degrees of $u$ and $v$ minus $2$ for the contracted edge $uv$. Degree of other vertices remains intact. Put $w$ into $K$, as $K$ contains vertices of degree $d-1$. This one vertex of $S$ was contracted to a vertex of $K$ and result stayed in $K$. Continue contracting while $S$ is non-empty.
  
  If $S$ is empty, we are left with a clique $K$. We can contract $K$ arbitrarily and the degree never exceeds $d-1$, because $K$ has $d$ vertices.
\end{proof}
\fi

\begin{proposition}
\label{prop:mcc}
  \BSAT (and hence \BWMC) is $\W[1]$\hy hard
  when parameterized by $k$ plus the signed twin-width of the
  incidence graph, even assuming that an optimal contraction sequence
  for the incidence graph is provided on the input.
\end{proposition}
\begin{proof}
  A graph $G$ is \emph{balanced $d$-partite} if $V(G)$ is the disjoint
  union of $d$ sets $V_1,\dots,V_d$, each of the same size, such that
  no edge of $G$ has both endpoints in the same set $V_i$.    The following problem is
  well-known to be $\W[1]$\hy complete when parameterized by $d$~\cite{Pietrzak03}.
\begin{center}
\vspace{-0.2cm}
\begin{boxedminipage}{0.98 \columnwidth}
\textsc{Partitioned Clique}\\[5pt]
\begin{tabular}{p{0.08 \columnwidth} p{0.8 \columnwidth}}
    Input: & A balanced $d$\hy partite graph $G$.\\
Question: & Does $G$ have a $d$\hy clique?
\end{tabular}
\end{boxedminipage}
\vspace{-0.2cm}
\end{center}

  Let $G$ be a balanced $d$\hy partite graph for $d\geq 2$ with $d$\hy
  partition $V_1,\dots,V_d$.  We write $V_i=\{v_1^i,\dots,v_n^i\}$,
  for $1\leq i \leq d$.  We construct a CNF formula $F$ which has a
  satisfying assignment that sets $\leq d$ variables to true if and
  only if $G$ has a $d$\hy clique.  As the variables of $F$ we take
  the vertices of~$G$.  $F$ consists of $d + |E(G)|$ many clauses: For
  each $1\leq i \leq d$ we add the clause $C_i=V_i$, which asserts
  that at least one variable from each set $V_i$ must be set to true.
  Each satisfying assignment that sets at most $d$ variables to true
  will therefore set exactly one variable from each $V_i$ to true,
  $1\leq i \leq d$.  
  For each pair of vertices $u,v$ with $u\in V_i$
  and $v\in V_j$ and $uv \notin E(G)$ we add the clause $C_{u,v}=(V_i\setminus \{u\}) \cup (V_j\setminus \{v\}) \cup \{\neg
    u, \neg v\} \}$,
  which asserts that for each satisfying assignment that sets at most
  $d$ variables to true, cannot set $u$ and $v$ to true if
  $uv\notin E(G)$.  We conclude that, indeed $F$ has a satisfying
  assignment that sets $\leq d$ variables to true if and only if $G$
  has a $d$\hy clique.

  It remains to show that the twin-width of the
  (unsigned) incidence graph $I$ of $F$ is bounded by a function of
  $d$. We observe that each set $V_i$ and each set
  $F_{i,j}=\SB C_{u,v} \SM u\in V_i, v\in V_j, uv\notin E(G)\SE$
  is
  a module of $I$. Hence we can contract each of these modules into a
  single vertex, obtaining a new graph $I'$ which is a subdivision of $K_d$. By Lemma~\ref{lemma_clique_subdivision}, $I'$ has (signed) twin-width at most $d-1$. 
\end{proof}

\section{Experiments}
In our experiments, we compute the signed twin-width, signed
clique-width, and treewidth of several SAT instances generated with
the tool CNFGen\footnote{\url{https://massimolauria.net/cnfgen/}} and
100 instances from the uniform random benchmark set
\emph{uf20-91}\footnote{\url{https://www.cs.ubc.ca/~hoos/SATLIB/benchm.html}}.
We used the respective
SAT-encodings~\cite{HeuleSzeider15,Parlak16,SamerVeith09,SchidlerSzeider22}
to compute the widths for the instances.  Signed clique-width and
treewidth were chosen because of being prominent and general examples
of structural parameters which yield the fixed-parameter tractability
of SAT and model
counting~\cite{SamerSzeider21}.%
  
We adapted the encoding for twin-width to signed graphs of formulas and added an improvement that can compute the signed twin-width of formulas with larger incidence graphs. %
In particular, the SAT encoding uses variables and clauses to
represent all possible contraction steps.
With bipartite contraction sequences,
it is possible to reduce the total number of admissible steps significantly. On a high level, we observed a significant performance improvement due to not needing to consider all possible contraction steps.

We used two types of instances in our experiments: hard instances from
\emph{proof complexity} and \emph{random $k$\hy SAT} instances.  We
generated all families of formulas available with CNFGen such that the
resulting signed graph had at most 125 vertices.  Where available, we
varied the parameters to see how the widths develop with the change in
the parameter.  We generated most of the random $k$-SAT instances with
CNFGen; all of these instances use 15 variables and a varying number
of literals per clause and clauses.  For each choice of parameters, we
generated 10 different instances and report the average.  Apart from
our generated instances, we use the publicized instance set uf20-91,
which contains 1000 random 3-SAT instances with 20 variables and 91
clauses, of which we randomly selected 100 instances.

Table~\ref{tab:exp_results} shows the results of our experiments.
For signed twin-width, we start from a greedy upper bound and can therefore provide
an upper bound whenever we cannot compute or verify the exact value of signed twin-width.

The results show that the signed twin-width never exceeds the treewidth.
Interestingly, it turned out that computing the signed clique-width using the best known SAT-encoding~\cite{Parlak16} is much harder than computing the signed twin-width.
Therefore, we can compare the two widths only on very few instances, and the
signed twin-width is smaller on all these instances.

\begin{table}
\caption{Experimental results comparing signed twin-width ($\tww$), signed clique-width ($\cw$) and treewidth ($\tw$).
Random instances are named such that the first number is the number
of variables, the second the number of literals per clause, and the third
the number of clauses. For random instances, each row represents a set of 10 generated instances.
Signed twin-width values marked with an \emph{*} are only an upper bound.}\label{tab:results}
\centering
\tiny
\begin{tabular}{@{}ll r r r r r@{}}
\toprule
&Instance&$\Card{V}$&$\Card{E}$&Signed $\tww$&Signed $\cw$&$\tw$\\
\midrule
\multirow{18}{*}{\rotatebox[origin=c]{90}{Proof Complexity}}
&cliquecolor5&130&225&*6.0&-&-\\
&count4&22&36&3.0&7.0&4.0\\
&count5&45&80&4.0&-&7.0\\
&count6&81&150&*5.0&-&10.0\\
&matching4&22&36&3.0&7.0&4.0\\
&matching5&45&80&4.0&-&7.0\\
&matching6&81&150&*5.0&-&10.0\\
&order4&46&96&5.0&-&9.0\\
&order5&95&220&*7.0&-&-\\
&parity5&45&80&4.0&-&7.0\\
&parity6&81&150&*5.0&-&10.0\\
&parity7&133&252&*6.0&-&-\\
&pidgeon4&34&48&4.0&6.0&4.0\\
&pidgeon5&65&100&*4.0&-&5.0\\
&subsetcard9&121&252&*6.0&-&-\\
&tseitin9&90&288&4.0&-&7.0\\
&tseitin10&100&320&4.0&-&9.0\\
&tseitin11&110&352&4.0&-&7.0\\
\midrule
\multirow{36}{*}{\rotatebox[origin=c]{90}{$k$\hy Random}}
&uf15\_2\_15&28.0&30.0&*3.0&4.0&2.0\\
&uf15\_2\_30&45.0&60.0&*3.8&6.2&4.0\\
&uf15\_2\_45&60.0&90.0&*4.0&-&6.3\\
&uf15\_2\_65&80.0&130.0&*4.2&-&7.6\\
&uf15\_2\_80&95.0&160.0&*5.0&-&8.5\\
&uf15\_2\_95&110.0&190.0&*5.0&-&9.4\\
&uf15\_2\_110&125.0&220.0&*5.0&-&9.8\\
\cmidrule{2-7}
&uf15\_3\_15&30.0&45.0&3.4&7.1&4.9\\
&uf15\_3\_30&45.0&90.0&*4.9&-&8.0\\
&uf15\_3\_45&60.0&135.0&*5.0&-&10.0\\
&uf15\_3\_65&80.0&195.0&*6.0&-&11.5\\
&uf15\_3\_80&95.0&240.0&*7.0&-&11.7\\
&uf15\_3\_95&110.0&285.0&*7.0&-&12.5\\
&uf15\_3\_110&125.0&330.0&*7.2&-&12.5\\
\cmidrule{2-7}
&uf15\_5\_15&30.0&75.0&5.0&11.4&9.6\\
&uf15\_5\_30&45.0&150.0&6.9&-&12.3\\
&uf15\_5\_45&60.0&225.0&*7.9&-&13.2\\
&uf15\_5\_65&80.0&325.0&*9.0&-&13.8\\
&uf15\_5\_80&95.0&400.0&*9.3&-&14.0\\
&uf15\_5\_95&110.0&475.0&*10.0&-&14.0\\
&uf15\_5\_110&125.0&550.0&*10.0&-&14.0\\
\cmidrule{2-7}
&uf15\_7\_15&30.0&105.0&6.0&14.1&12.6\\
&uf15\_7\_30&45.0&210.0&8.0&-&13.8\\
&uf15\_7\_45&60.0&315.0&*9.0&-&14.0\\
&uf15\_7\_65&80.0&455.0&*10.0&-&14.0\\
&uf15\_7\_80&95.0&560.0&*11.0&-&14.0\\
&uf15\_7\_95&110.0&665.0&*11.0&-&14.0\\
&uf15\_7\_110&125.0&770.0&*11.5&-&14.0\\
\cmidrule{2-7}
&uf15\_10\_15&30.0&150.0&6.7&15.3&13.9\\
&uf15\_10\_30&45.0&300.0&9.0&-&14.0\\
&uf15\_10\_45&60.0&450.0&*10.0&-&14.0\\
&uf15\_10\_65&80.0&650.0&*11.0&-&14.0\\
&uf15\_10\_80&95.0&800.0&*12.0&-&14.0\\
&uf15\_10\_95&110.0&950.0&*12.0&-&14.0\\
&uf15\_10\_110&125.0&1100.0&*12.0&-&14.0\\
\cmidrule{2-7}
&uf20\_3\_91&111.0&273.0&*7.7&-&14.5\\
\bottomrule
\end{tabular}
  \label{tab:exp_results}
\end{table}

\section{Concluding Remarks}
We have provided an exhaustive investigation of how twin-width can be used in SAT solving and model counting and have developed the notion of signed twin-width for formulas. Our complexity-theoretic results follow up on the classical line of research that investigates the complexity of SAT and its extensions from the viewpoint of variable-clause interactions. On the empirical side, we have computed the exact signed twin-width of several formulas and compared these values to those of treewidth and signed clique-width.

In future work, it would be interesting to investigate whether there
is a structural parameter that generalizes signed clique-width and can
yield fixed-parameter algorithms for SAT alone; in the case of
twin-width, we show that this is not possible
(Proposition~\ref{prop:paranp}).  Moreover, while the fixed-parameter
tractability established in Theorem~\ref{the:main} should be viewed as
a classification result, it would be interesting to see whether the
ideas developed there could be used to inspire improvements to
existing heuristics for SAT.  \bibliographystyle{plainurl}
\bibliography{biblio}
\end{document}